\documentclass{book}
\usepackage[utf8]{inputenc}

\title{Machine and Deep Learning for Credit Scoring: A compliant approach}
\author{Abdollah RIDA}
\date{July $2$, 2019}

\usepackage{natbib}
\usepackage{graphicx}
\usepackage{amsmath,amsthm,amssymb}
\usepackage{algorithm}
\usepackage[noend]{algpseudocode}
\usepackage{float}

\usepackage[margin=1in]{geometry}

\usepackage[toc,page]{appendix}

\usepackage{mathtools}

\DeclarePairedDelimiter\abs{\lvert}{\rvert}%

\theoremstyle{definition}
\newtheorem{definition}{Definition}

\theoremstyle{plain}
\newtheorem{theorem}{Theorem}

\newtheorem{propriety}{propriety}

\makeatletter
\def\BState{\State\hskip-\ALG@thistlm}
\makeatother

\DeclareMathOperator{\argmin}{argmin}

\usepackage{lipsum}
\newcommand\abstractname{Abstract}  %%% here
\makeatletter
\if@titlepage
  \newenvironment{abstract}{%
      \titlepage
      \null\vfil
      \@beginparpenalty\@lowpenalty
      \begin{center}%
        \bfseries \abstractname
        \@endparpenalty\@M
      \end{center}}%
     {\par\vfil\null\endtitlepage}
\else
  \newenvironment{abstract}{%
      \if@twocolumn
        \section*{\abstractname}%
      \else
        \small
        \begin{center}%
          {\bfseries \abstractname\vspace{-.5em}\vspace{\z@}}%
        \end{center}%
        \quotation
      \fi}
      {\if@twocolumn\else\endquotation\fi}
\fi
\makeatother

\begin{document}

\maketitle

\newpage

\tableofcontents

\listoffigures
 
\listoftables

\begin{abstract}
Credit Scoring is one of the problems banks and financial institutions have to solve on a daily basis. If the state-of-the-art research in Machine and Deep Learning for finance has reached interesting results about Credit Scoring models, usage of such models in a heavily regulated context (\citep{FED1}, \citep{FED2}, \citep{OCC1}, \citep{OCC2}) such as the one in banks has never been done so far. 

Our work is thus a tentative to challenge the current regulatory status-quo and introduce new BASEL 2 and 3 compliant techniques, while still answering the Federal Reserve Bank and the European Central Bank requirements.

With the help of Gradient Boosting Machines (mainly XGBoost \citep{DBLP:journals/corr/ChenG16}) we challenge an actual model used by BANK A for scoring through the door Auto Loan applicants. We prove that the usage of such algorithms for Credit Scoring models drastically improves performance and default capture rate.

Furthermore, we leverage the power of Shapley Values \citep{DBLP:journals/corr/LundbergL17} to prove that these relatively simple models are not as black-box as the current regulatory system thinks they are, and we attempt to explain the model outputs and Credit Scores within the BANK A Model Design and Validation framework.
\end{abstract}

\chapter{Introduction}
During US subprime mortgage crisis and the European sovereign debt crisis, established financial institutions in the USA and Europe suffered huge losses. The crisis, mainly due to misuse of credit default swaps (CDS), raised concenrns about credit risk. Credit risk management is becoming an increasingly important factor and attracted significant attention from researchers and market participants. In order to effectively manage the credit risk exposures, optimize capital, answer to regulation and increase profits, financial institutions switched their focus to developing an accurate credit scoring model.

After the introduction of the commercial scorecard, many statistical methods have been used for credit risk assessment. Despite their wide application, these models cannot capture the complex financial relations specific to credit risk. Related studies \citep{saberi_granular_2013} have shown that machine learning techniques are superior to that of statistical techniques in dealing with credit scoring problems. Currently, simple models such as Logistic regression or simple decision trees are the most frequently used statistical models. Meanwhile, some shallow architectures such as support vector machines (SVMs) and multi-layer perceptron (MLPs) with a single hidden layer, have been applied to this problem \citep{bhatia_credit_2017} \citep{khandani_consumer_2010}.

Shallow architectures have been shown effective in solving many simple or well-constrained problems. However, these methods mainly focus on the outputs of classifiers at the abstract level, while neglecting the rich information hidden in the confidence degree. Their limited modeling and representational power can cause difficulties when dealing with more complicated real-world applications.

On the other hand Gradient Boosting, and specifically XGBoost \citep{DBLP:journals/corr/ChenG16}, has shown promising results on a multitude of real world problems. Beyond being the recommended algorithm for creating effective and reliable baselines for any machine learning problem, XGBoost is: blazingly fast; handles missing values without the use of imputation; captures non-linearities effectively. He is thus, among all other Gradient Boosting Algorithm candidates, the best choice given our computational constraints. Moreover, the python package comes natively equipped with the possibility of using Shapley Values \citep{DBLP:journals/corr/LundbergL17} to explain the model outputs and provide a more in-depth understanding of how the scoring process is done. This combination allows us to comply to the FED's \citep{FED1} \citep{FED2} and OCC's \citep{OCC1} \citep{OCC2} requirements from credit scoring models.

To our knowledge, this is the first comprehensive study of Gradient Boosting Models in corporate and retail/wholesale credit rating based on real bank data. Therefore, this memoir fills in such a literature gap by introducing XGBoost as the algorithm for credit rating to generate fast and accurate individual classification and scoring results. The goal is to provide a set of descriptive results and tests that lay a foundation for future theoretical and empirical work on XGBoost in credit scoring in auto loan markets, but also for corporate lending. In this memoir, we investigate the performances of different credit scoring models by conducting experiments on a collection of auto loan data. 

The remainder of the memoir is organized as follows. Chapter 2 describes the model framework and theory examined. Chapter 3 describes the experiments performed. Chapter 4 presents the empirical results from comparing our model to BANK A's. Chapter 5 focuses on explainability. Appendices are provided to introduce the imbalance problem and provide proofs for theorems we used.

\chapter{Model Framework and Theory}

In this section we describe the theory behind the tree based models, boosting, cross-validation and class rebalancing. We also provide proofs that guided our modelling decisions.

\section{Mathematical notions}

\subsection{Cross-validation, Class weights and overview of the algorithm}

Before delving into the mathematical intricacies of cross-validation and class reweighting, let us first correctly define the mathematical setting we are working in.

Let $(\Omega, \mathcal{F}, \mathbb{P})$ be a probability space. Assume that $(X,Y)$ is a couple of random variables defined on $(\Omega, \mathcal{F}, \mathbb{P})$ and taking values in $\chi \times \{-1,1\}$ where $\chi$ is a given state space (the bank, behavioral and account features in our case). Our model's aim is to define a function $h: \chi \longrightarrow \{-1,1\}$ called classifier such that $h(X)$ is the best prediction of $Y$ in a given context. For instance, the probability of misclassification of $h$ is:

$$L_{miss}(h) = \mathbb{P}\left ( Y \neq h(X) \right )$$

Note that $\mathbb{E}\left[X|Y\right]$ is a random variable measurable with respect to the $\sigma-$algebra $\sigma(X)$. Therefore there exists a function $\eta: \chi \longleftarrow \left[-1,1\right]$ so that $\mathbb{E}\left[X|Y\right] = \eta(X)$ almost surely. The following theorem is a well-known result:

\begin{theorem}

The Bayesian classifier $h_*$ defined for all $x \in \chi$ by:

\[
    h_*(x)= 
\begin{cases}
    1,& \text{if } \eta(x)\ge 0\\
    -1,              & \text{otherwise}
\end{cases}
\]

is such that:

$$h_* = \argmin_{h: \chi \longrightarrow \{-1,1\}} L_{miss}(h)$$

\end{theorem}

\begin{proof}
  Let us consider the following minimization problem:
  $$\argmin_{h: \chi \longrightarrow \{-1,1\}} L_{miss}(h)$$
  The problem is equivalent to:
  $$\argmin_{h: \chi \longrightarrow \{-1,1\}} \mathbb{E}_X\left[\mathbb{E}_{Y|X}\left[\mathbf{1}_{Y_i \neq h(X)}\right]\right]$$
  Then:
  \begin{equation*}
    \begin{aligned}
    \mathbb{E}_X\left[\mathbb{E}_{Y|X}\left[\mathbb{1}_{Y_i \neq h(X)}\right]\right] &= \mathbb{E}_X\left[\mathbb{P}(Y=1|X) - \mathbb{P}(Y=-1|X)\right] \\
    & = \mathbb{E}_X\left[\mathbf{1}_{1 \neq h(X)}\mathbb{P}(Y=1|X) - \mathbf{1}_{-1 \neq h(X)} \mathbb{P}(Y=-1|X)\right] \\
    & - \mathbb{E}_X\left[-\mathbf{1}_{-1 \neq h(X)}\mathbb{P}(Y=1|X) + \mathbf{1}_{1 \neq h(X)} \mathbb{P}(Y=-1|X)\right]  \\
    & = \mathbb{E}_X\left[\mathbf{1}_{1 \neq h(X)}\mathbb{E}\left[Y|X\right]\right] - \mathbb{E}_X\left[\mathbf{1}_{-1 \neq h(X)}\mathbb{E}\left[Y|X\right]\right]
    \end{aligned}
  \end{equation*}
  Thus to minimize the above expression $h$ must verify:
  \[
    h(x)= 
  \begin{cases}
    1,& \text{if } \mathbb{E}\left[Y|X\right]\ge 0\\
    -1,              & \text{otherwise}
  \end{cases}
  \]
\end{proof}

However in practice the minimization of $L_{miss}$ holds on a specific on a specific set $H$ f classifiers (weak decision trees in our case) which may possibly not contain the Bayes classifier. Moreover, since in most cases the classification risk $L_{miss}$ cannot be computed nor minimized, it is instead estimated by the empirical classification risk defined as:

$$\hat L_{miss}^n(h)= \frac{1}{n} \sum_{i=1}^{n} \mathbf{1}_{Y_i \neq h(X_i)}$$

Where $(X_i,Y_i)$ are independent observations with the same distribution as $(X,Y)$. The classification problem then boils down to solving:

$$\hat h^n_H = \argmin_{h \in H} \hat L_{miss}^n(h)$$

Using Hoeffding's inequality, we can easily prove\footnote{A quick proof is provided in the appendix} that when $H = \{h_1, ..., _M\}$ then for all $\delta \ge 0$:

$$\mathbb{P} \left[L_{miss}\left ( h^n_H\right ) \leq \min_{1 \leq j \leq M} L_{miss} \left (h_j \right ) + \sqrt{\frac{2}{n}\log\left (\frac{2M}{\delta} \right )} \right] \geq 1 - \delta$$

Giving us a “confidence bound” for our predictor. In practice, this translates to:
\begin{itemize}
    \item The more data the better the prediction
    \item The bigger the predictor set , the bigger the error
    \item A higher confidence level  leads to a bigger error
\end{itemize}

\subsection{Cross-validation}

Given the mathematical context above, we can now develop the theory behind cross-validation. We have seen that a large dictionary of predictors $H$ and a large quantity of data gives us a better prediction, but in our case the bigger the $H$ the higher the risk of overfitting becomes since we are randomly generating trees. That is why we need cross-validation.
The goal of cross-validation is to assess the quality of a given machine learning method. It computes error estimates on training and validation sets to choose the most promising ones. One can see cross-validation as an estimate of the average risk of a Machine Learning method. It does not yield an error bound on the predictor obtained in practice. The idea is that once we have identified our best combination of parameters we test the performance of that set of parameters in a different context. That is k-fold cross-validation.
In k-fold cross-validation, the original sample is randomly partitioned into k equal sized subsamples. Of the k subsamples, a single subsample is retained as the validation data for testing the model, and the remaining $k-1$ subsamples are used as training data. The cross-validation process is then repeated k times, with each of the k subsamples used exactly once as the validation data. The k results can then be averaged to produce a single estimation.

Mathematically this can be modelled the following way: Let $\kappa : \{1,...,N\} \longrightarrow \{1,...,N\}$ be an indexing function that indicates the partition to which observation $i$ is allocated by the randomization. Denote by $\hat f^{-k}(x)$ the fitted function, computed with the $k^{th}$ part of the data removed. Then the cross-validation estimate of prediction error is:

$$CV(\hat f) = \frac{1}{N} \sum_{i = 1}^n L \left (y_i, \hat f^{-\kappa (i)}(x_i)\right )$$

Given a set of models $f(x, \alpha)$ indexed by a tuning parameter $\alpha$, denote by $f^{-k}(x, \alpha)$ the $\alpha^{th}$ model fit with the $k^{th}$ part of the data removed. Then for this set of models we define:

$$CV(\hat f, \alpha) = \frac{1}{N} \sum_{i = 1}^n L \left (y_i, \hat f^{-\kappa (i)}(x_i, \alpha)\right )$$

The function $CV(\hat f, \alpha)$ provides an estimate of the test error curve, and we find the tuning parameter $\hat \alpha$ that minimizes it. Our final chosen model is $f(x, \hat \alpha)$ which we will then fit to all the data. Plotting the evolution of the cross-validation error and the training error at the same time can also give us more insight about whether the model overfits or not.

\subsection{Loss Reweighting}

Class imbalance\footnote{A more detailed study of imbalance and how it can hurt your models is provided in appendix} is an issue that is common in credit portfolios: the good borrowers are much more present than the defaulting ones. This makes the classifiers too attracted to the majority class and as a result makes training and testing errors not the same. We can use class resampling or loss reweighting to solve this issue. Loss reweighting rewrites our loss function as:

$$C(Y) \, l(Y, h(X))$$

Where C puts more emphasis on some classes than others. This means that if our testing error target is:

$$\mathbb{E}_{\pi_t}\left[C_t(Y) \, l(Y, h(X))\right] = \sum_k \pi_t(k) \, C_t(k) \, \mathbb{E}\left[l(Y, h(X))|Y=k\right]$$

And our training error is:

$$\mathbb{E}_{\pi_{tr}}\left[C_{tr}(Y) \, l(Y, h(X))\right] = \sum_k \pi_{tr}(k) \, C_{tr}(k) \, \mathbb{E}\left[l(Y, h(X))|Y=k\right]$$

Where $\pi$ is the class probability. To make the errors the same we can combine resampling and loss weighting by choosing:

$$C_{tr}(k) = C_t(k) \, \frac{\pi_t(k)}{\pi_{tr}(k)}$$

It is worth noting that this induces a change of probability measure as all expectations are computed under the new reweighted measure. This means that the model does not predict real probabilities (as in probabilities under the original probability measure) anymore.

\section{Prediction Model}

\subsection{Decision Trees}

The tree ensemble model consists of a set of classification and regression trees (CART). Here’s a simple example of a CART that classifies whether someone will like a hypothetical computer game X in Figure 2.1.

\begin{figure}[htbp]
\centering
\includegraphics[width=8cm]{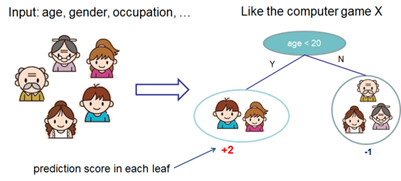}
\caption{A CART Decision Tree from \citep{DBLP:journals/corr/ChenG16}}
\end{figure}

In these CARTs, a real score is associated with each of the leaves, which gives us richer interpretations that go beyond classification. This also allows for a principled, unified approach to optimization.
Usually, a single tree is not strong enough to be used in practice. What is actually used is the ensemble model is a forest, which sums the prediction of multiple trees together. 

\begin{figure}[htbp]
\centering
\includegraphics[width=8cm]{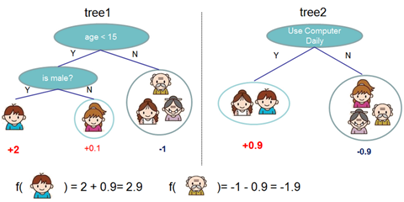}
\caption{Two XGBoost trees with score as output (instead of prediction) \citep{DBLP:journals/corr/ChenG16}}
\end{figure}

\subsubsection{Boosting}

The idea behind boosting is to learn a sequence of weak predictors trained on a weighted dataset with the weights depending on the loss so far. 

\begin{figure}[h]
\centering
\includegraphics[width=8cm]{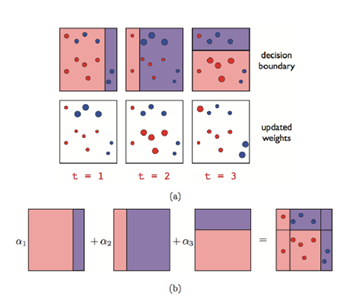}
\caption{Boosting process for a decision tree}
\end{figure}

A weak learner is usually a simple predictor that is: 1) easy to learn, 2) only needs to be slightly better than a constant predictor. Some examples of weak learners are: Decision trees with few splits, Stumps (Decisions trees with one split) and Generalized Linear Regressions with few variables.
The boosting process is therefore just a sequential linear combination of weak learners that attempts to minimize a loss.

Formally, a CART splits the space of all joint predictor variable values into disjoint partitions $R_j$ where $j$ represents the terminal node of the tree. A constant $\gamma_j$ is assigned to each partition and the predictive rule is:

$$x \in R_j \Longrightarrow f(x) = \gamma_j$$

Thus a tree can be formally expressed as:

$$T(x; \Theta) = \sum_{j=1}^J \gamma_j \, \mathbf{1}_{x \in R_j}$$

With parameters $\Theta = {R_j, \gamma_j}_{1 \leq j \leq J}$. J is usually treated as a meta-parameter. The parameters are found by minimizing the empirical risk:

$$\hat \Theta = \argmin_{\Theta} \sum_{j=1}^J \sum_{x_i \in R_j} L(y_i, \gamma_j)$$

As specified in \citep{hastie01statisticallearning}, this combinatorial optimization problem is complicated, and we usually settle for approximate suboptimal solutions. The problem can be divided in two parts: Finding $\gamma_j$ given $R_j$ and findng $R_j$.

We previously described a general strategy to find the best classifier and \citep{hastie01statisticallearning} describes one for classification trees. In our case the Gini (or AUROC) replaces the misclassification loss in the growing of the tree. The boosted tree model is thus the sum of such trees:

$$f_M(x) = \sum_{m=1}^M T(x; \Theta_m)$$

Induced in a forward stagewise manner\footnote{The forward stagewise algorithm can be found in the appendix}. At each step in the forward stagewise procedure one must solve:

$$\hat \Theta_m = \argmin_{\Theta_m} \sum_{i=1}^N L\left (y_i, f_{m-1}(x_i) + T(x_i; \Theta_m)\right )$$

\subsection{Gradient Boosting}

The following presents the generic gradient tree-boosting algorithm for regression. Specific algorithms are obtained by inserting different loss criteria $L(y, f(x))$. The first line of the algorithm initializes to the optimal constant model, which is just a single terminal node tree. The components of the negative gradient computed at line 4 are referred to as generalized or pseudo residuals, $r$.

\begin{algorithm}
\caption{Gradient Tree Boosting Algorithm}\label{GTB}
\begin{algorithmic}[1]
\Procedure{Initialize}{}
\State $f_0(x) \gets \argmin_{\gamma} \sum_{i = 1}^N L(y_i,\gamma)$
\BState \emph{For $m = 1$ to M}:
\State For $i = 1,...,N$ compute $r_{i,m} = - \left[\frac{\partial L(y_i, f(x_i))}{\partial f(x_i)}\right]_{f = f_{m-1}}$.
\State Fit a regression tree to targets $r_{i,m}$ giving terminal partitions $R_{j,m}$, $j=1,...,J_m$
\State For $j=1,...,J_m$ compute $\gamma_{j,m} = \argmin_{\gamma} \sum_{x_i \in R_{j,m}} L(y_i, f_{m-1}(x_i) + \gamma)$
\State Update $f_m(x) = f_{m-1}(x) + \sum_{j=1}^{J_m} \gamma_{jm} \mathbf{1}_{x \in R_{jm}}$
\EndProcedure
\Procedure{Output}{}
\State Return $\hat f(x) = f_M(x)$
\EndProcedure
\end{algorithmic}
\end{algorithm}

The algorithm for classification is similar. Lines 3–7 are repeated K times at each iteration m, once for each class. The result at line 9 is K different (coupled) tree expansions $f_{k,M}(x), \, k = 1,...,K$. These produce probabilities or do classification \citep{hastie01statisticallearning}.

\subsection{XGBoost}

XGBoost \citep{DBLP:journals/corr/ChenG16} pushes Gradient Boosting to the limit. Gradient Boosting carries the principle of Gradient Descent and Boosting to supervised learning. Gradient Boosted Models (GBMs) are trees built sequentially; XGBoost is parallelized and is thus blazingly faster.

\begin{itemize}
\item Each new model uses Gradient Descent optimization to update/ make corrections to the weights to be learned by the model to reach a local minimum of the cost function.
\item The vector of weights assigned to each model is derived from the weights optimized by Gradient Descent to minimize the cost function. The result of Gradient Descent is the same function of the model as the beginning, just with better parameters.
\item Gradient Boosting adds a new function to the existing function in each step to predict the output. The result of Gradient Boosting is an altogether different function from the beginning, because the result is the addition of multiple functions.
\end{itemize}

XGBoost’s objective function is as follows: 

$$obj = \sum_{i=1}^n l(y_i, \hat y_i^{(t)}) + \sum_{i=1}^t \Omega(f_i)$$

Where $\Omega$ is the regularization term that controls model complexity and prevents overfitting. In the above equation $f_i$ is the tree number $i$ and $y_i$ (resp. $\hat y_i$) the true class (resp. the predicted class) of the $i^{th}$ data point. $l$ is the loss function.

\section{Model Training, Calibration and validation}

The performance measures used for training and cross-validation are \textbf{ROCAUC} and \textbf{Log-Loss}, while \textbf{$\mathbf{F_{\mathbf{\beta}}}$-score} was used for hyper-parameter tuning. Log-loss (also known as Binary Cross-Entropy) measures absolute probabilistic difference between our predicted classes. It drives the model to make predictions that are better at separating the two classes. It is given by the following formula:

$$-{(y\log(p) + (1 - y)\log(1 - p))}$$

While $F_{\beta}$-score measures the effectiveness of information retrieval for recall and  times precision (i.e. we attach $\beta^2$ times more importance to precision than recall). It is given by:

$${\displaystyle F_{\beta }=(1+\beta ^{2})\cdot {\frac {\mathrm {precision} \cdot \mathrm {recall} }{\beta ^{2}\cdot \mathrm {precision} +\mathrm {recall} }}}$$

As a reminder, precision and recall are given by:

$$precision = \frac{true \, positives}{true \, positives + false \, positives}$$

$$recall = \frac{true \, positives}{true \, positives + false \, negatives}$$

\chapter{Model Specifications and Estimation}

Every Machine Learning project can be divided into 3 major steps: data preprocessing and feature engineering, model building and training and finally model calibration, fine tuning and validation.
Figure 3.1 shows an example of the data preparation/preprocessing pipeline that is usually used for Machine Learning projects

\begin{figure}[h]
\centering
\includegraphics[width=16cm]{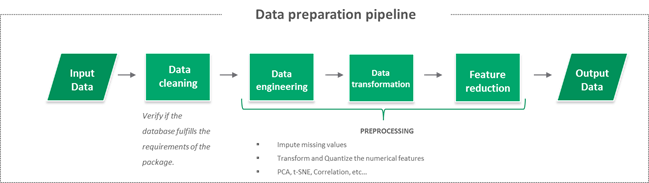}
\caption{Data Preparation Pipeline}
\end{figure}

\section{Model Specifications}
\subsection{Target Variable}

The model target was defined as ‘Ever 60 day past due or worse, including charge-offs and repossessions, within the first 18 months of origination.’ Development records meeting this definition were assigned a target value of 0 (Bad). No indeterminate performance was used. All development applicants not assigned a target value of 0, received a target value of 1 (Good). For BANK A booked records, the performance flag provided on the input field was used. 

\subsection{Primary Modeling Technique}

The preprocessing used for the dataset is light. There are three main axes: missing values, categorical feature encoding and variable selection. All variables go through a process of encoding using Weight of Evidence. Measures such as KS, PR, and ROC are compared between the training and test development samples, as well as the out-of-time validation population, to ensure the model did not over-fit the training sample. Close values between samples indicated the model validates. Score distributions are also exampled to ensure a smooth rank ordering of risk and a probabilistic tendency to separate the two classes.
For the BANK A auto prime model development, as a preliminary step, a Gradient Boosted Model (GBM) with a depth of six was initially run to develop a quick baseline expectation of best-case model performance (i.e. likely highest KS attainable) prior to constraining for fair credit reporting. Several model versions were built with different parameters and hyper-parameters. For further detail on these iterations, please see section Key Modeling Decisions and Alternative Specifications.

\subsection{Missing Values}

No missing values imputation was done. XGBoost handles missing values the following way: at each node of the decision tree, IF the tested value is missing THEN redirect towards the default value (which is specified at each node).
We tried building models with missing value imputation (replacing with a default value or simply removing observations with missing values) but we noticed that the best performing model is the one with no missing imputation.

\subsection{Encoding Categorical Features}

We used Weight of Evidence (WOE) encoding for categorical features. WOE is given by:

$$WOE_c = \log \left[ \frac{\frac{Goods_c}{Goods_{total}}}{\frac{Bads_c}{Bads_{total}}}\right] \times 100$$

Bad ratios will have a WOE less than zero, while good ratios will be greater than zero. Those intervals with a WOE near zero will be neutral. We used the scikit-learn contrib library categorical-encoders that include WOE encoding. Correspondence between categorical values and their encoding is accessible by calling the encoder.transform method and examine the relevant features before and after encoding.

\subsection{Variable Selection}

We removed all the features that were part of Credit Bureau B’s model output, as well as all external and bureau scores. We also removed all features that we believed would violate the Fair Credit Reporting Act. 
We tried using several variable selection methods, but in the end we opted for XGBoost’s feature importance and the underlying Shapley score \citep{DBLP:journals/corr/LundbergL17} of each feature to do feature reduction. The model thus has no variable selection/feature reduction before training.

\section{Model Development Tools}

Model developed in Python 3.7. Datasets provided from BANK A shared in .csv format.

\subsection{Main Dependencies}

The main requirements are:
numpy, pandas scikit-learn, xgboost, category-encoders, joblib, tqdm

\subsection{Dependencies for Plots and Explainability}

These dependencies are optional (but are currently required by the code) to plot outputs and explain the model:
matplotlib, seaborn, scikitplot, shap

\section{Key Modeling Decision and Alternative Specifications}

The first model that was built was intended to actually predict the real probability of default and thus could be used as a decision-making model. Unfortunately, its class separation power was extremely low. Figure 3.3 shows the results of this particular model.
This calibration was done using the max\_delta\_step parameter. In \citep{DBLP:journals/corr/ChenG16}, the definition of the weight updating is the following:

$$w_j^* = -\frac{\sum_{i \in I_j} g_i}{\sum_{i \in I_j} h_i + \lambda}$$

\begin{figure}[h]
\centering
\includegraphics[width=16cm]{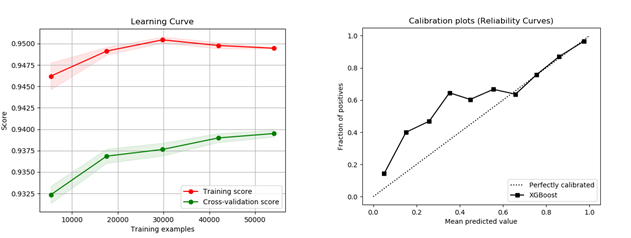}
\caption{Learning Curve and Reliability curve for the in-time dataset for the calibrated model}
\end{figure}

Where $h_i = \partial^2_{\hat y^{(t-1)}} l(y_i, \hat y^{(t-1)})$ the hessian (the reader can refer to \citep{DBLP:journals/corr/ChenG16} for original notations).
In the cases where our training dataset contains high imbalance, it is clear that the value of  will be extremely small for the minority class. Making the weight updates even more biased towards the majority class.
max\_delta\_step is a parameter that aims to bound the absolute value of the inverse of the  hessian matrix (i.e. our weight updates). Typical values are between 1 and 10. By introducing this parameter, we avoid the probability measure change induced by the class re-weighting (even if we can still use it) and are thus computing real world probabilities.

We plot the model’s performance during training and cross-validation in Figure 3.2 to prove that it does not over fit but also that it behaves similarly. We also plot reliability curves for this model that are useful for determining whether or not we can interpret the predicted probabilities directly as a confidence level. 

\begin{figure}[h]
\centering
\includegraphics[width=10cm]{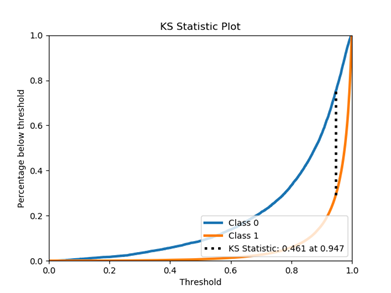}
\caption{KS Statistic plot for the out-of-time dataset for the calibrated model}
\end{figure}

As we can see, the model is extremely consistent and stable during training and cross-validation (that are done at the same time). It is also very well calibrated, especially for higher probabilities. Unfortunately the model’s separation power is really low.

In Figure 3.4 the confusion matrix poor performance is due to it being plotted for a 50\% threshold. The distribution plot shows once again that the model performs poor class separation.

\begin{figure}[h]
\centering
\includegraphics[width=16cm]{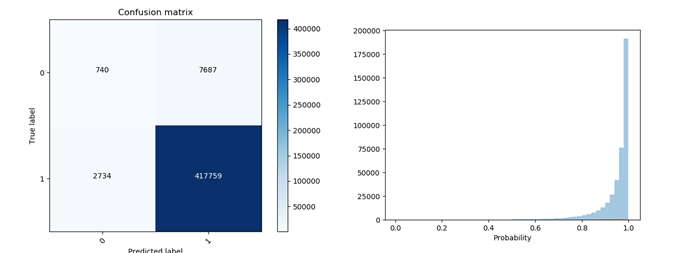}
\caption{Confusion Matrix and Score Distribution for the out-of-time dataset for the calibrated model}
\end{figure}

ROC and PR curves in Figure 3.5 show the model’s performance and prove that it’s hard to correctly separate the two classes without additional features. 

\begin{figure}[h]
\centering
\includegraphics[width=16cm]{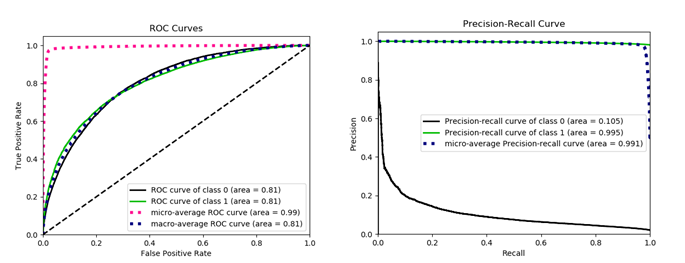}
\caption{ROC and PR Curves for the out-of-time dataset for the calibrated model}
\end{figure}

After the previous experiments, we decided to change the model’s parameters to more 	conservative ones to force class separation. After discussing with business as well as model 	owners we decided to output raw score (log-odds) as well as try the original thirteen variables used in the Credit Bureau B model. While this cannot assess the algorithm’s added value, since constraining 	the available features to the tree building algorithm will reduce the available tree space to one 	smaller that might not contain trees that can uncover the non-linearities present, it can show the power of the optimization algorithm that still manages to improve the results of the Credit Bureau B algorithm.

\begin{figure}[!h]
\centering
\includegraphics[width=8cm]{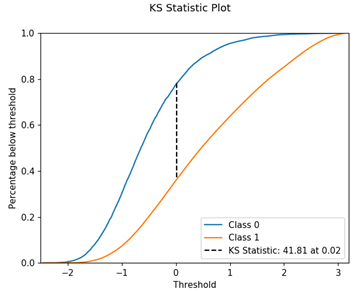}
\caption{KS Statistic Plot for the out-of-time dataset for the model using the original 13 variables}
\end{figure}

The distribution plot shows once again that the model performs poor class separation, yet that still outperforms that outputted by the Credit Bureau B model.	

\begin{figure}[!h]
\centering
\includegraphics[width=14cm]{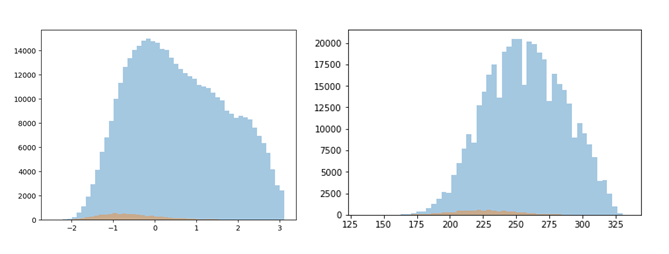}
\caption{Score distribution for our model using the original 13 variables (Left) VS BANK A's (Right) on the in-time dataset}
\end{figure}
	
\begin{figure}[!h]
\centering
\includegraphics[width=14cm]{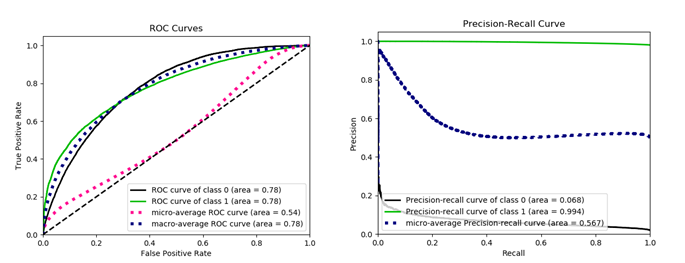}
\caption{ROC and PR on the out-of-time dataset for the model using the original 13 variables}
\end{figure}

ROC and PR curves show the model’s performance and prove that it’s hard to correctly separate the two classes without additional features. 

\pagebreak[4]

\section{Final Model Form and Specifications}

The final model was obtained by training the algorithm on the entire in-sample dataset and features, then appending the most important ones to the original thirteen features used in the Credit Bureau B model. The final model has the following parameters and hyper parameters. We provide an explanation below for the most meaningful and impactful ones:

\begin{figure}[ht]
\centering
\includegraphics[width=14cm]{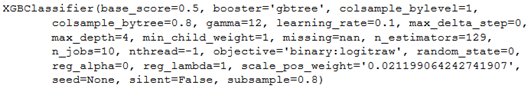}
\caption{Final model specifications}
\end{figure}

\begin{itemize}
    \item alpha: L1 regularization term on weights. Increasing this parameter will make model more conservative.
    \item lambda: L2 regularization term on weights. Increasing this value will make model more conservative.
    \item gamma: minimum score change to split a node. Is useful to help avoid overfitting. The larger gamma is, the more conservative the algorithm will be.
    \item max\_depth: maximum depth of a tree
    \item learning\_rate: Step size shrinkage used in update to prevents overfitting. After each boosting step, we can directly get the weights of new features, and eta shrinks the feature weights to make the boosting process more conservative.
    \item scale\_pos\_weight: Control the balance of positive and negative weights, useful for imbalanced classes. A typical value to consider: sum(negative instances) / sum(positive instances).
    \item max\_delta\_step: Maximum delta step we allow each leaf output to be. If the value is set to 0, it means there is no constraint. If it is set to a positive value, it can help making the update step more conservative. Usually this parameter is not needed, but it might help in logistic regression when class is extremely imbalanced: In extreme cases where the Hessian is nearly 0 (which is the case with imbalance) the weights of the majority class because infinite. This parameter helps by introducing an absolute regularization capping the weight.
    \item objective: binary:logistic: logistic regression for binary classification, output probability.
    \item min\_child\_weight: Minimum sum of instance weight (hessian) needed in a child. If the tree partition step results in a leaf node with the sum of instance weight less than min\_child\_weight, then the building process will give up further partitioning. The larger min\_child\_weight is, the more conservative the algorithm will be.
    \item colsample\_bytree: is the subsample ratio of columns when constructing each tree. Subsampling occurs once for every tree constructed.
    \item colsample\_bylevel: is the subsample ratio of columns for each level. Subsampling occurs once for every new depth level reached in a tree. Columns are subsampled from the set of columns chosen for the current tree.
    \item base\_score [default=0.5]: The initial prediction score of all instances, global bias. For sufficient number of iterations, changing this value will not have too much effect.
    \item subsample [default=1]: Subsample ratio of the training instances. Setting it to 0.5 means that XGBoost would randomly sample half of the training data prior to growing trees. and this will prevent overfitting. Subsampling will occur once in every boosting iteration.
\end{itemize}

\chapter{Model Performance and Testing}
\section{Overall Results}

We tested the methodology on a banking model for the autoloan portfolio. We used the Kolmogorov-Smirnov (KS) statistic as well as confusion matrices, Area under Receiver Operator and Precision-Recall Curves (AUROC) as metrics to measure performance. 

\begin{table}[h]
\begin{center}
\begin{tabular}{|l|c|c|c|c|c|c|}
\hline
Model     & KS in-time & KS OOT & AUROC in-time & AUROC OOT & PR in-time & PR OOT \\ \hline
Our Model & 47.8       & 44.91          & 0.81          & 0.80              & 0.093      & 0.093          \\ \hline
BANK A    & 41.89      & 41.31          & 0.77          & 0.77              & \_\_       & 0.06           \\ \hline
\end{tabular}
\caption {Comparison summary between our model and BANK A’s model} \label{tab:Comparison} 
\end{center}
\end{table}    

While the model has a relatively good early pick up, it quickly lags behind when it comes to correctly detecting default.
Finally, we plot the model’s performance during training and cross-validation to prove that it does not over fit but also that it behaves similarly. We also plot reliability curves for this model that are useful for determining whether or not we can interpret the predicted probabilities directly as a confidence level. 

\begin{figure}[ht]
\centering
\includegraphics[width=16cm]{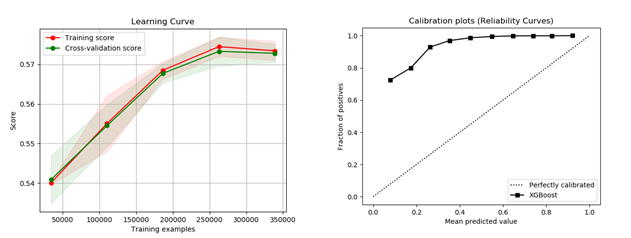}
\caption{Learning and Reliability Curves for the in-time dataset for the final model}
\end{figure}

As we can see, the model is extremely consistent and stable during training and cross-validation (that are done at the same time), it however needs calibration as the class imbalance (and the subsequent rebalancing) tends to make it difficult for the model to have reasonable confidence levels.

\section{Testing on Out-of-Time Data}

In addition to the out-of-sample plots, we also plot certain metrics for the original Credit Bureau B score. The model is also better than the original model built by Credit Bureau B at separating the two classes. This is thanks to using the logarithmic loss as an objective: The model aims to probabilistically separate defaults from good borrowers when scoring them, thus lowering the good rates in the high default bins.

\begin{figure}[h!]
\centering
\includegraphics[width=16cm]{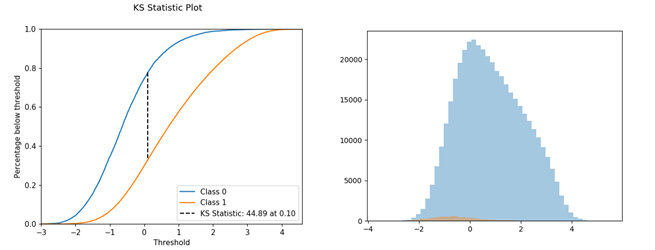}
\caption{KS statistic and score distribution for the final model on the out-of-time dataset for the final model}
\end{figure}

ROC and PR curves show the model’s performance and prove that it’s hard to correctly separate the two classes without additional features. 

\begin{figure}[h!]
\centering
\includegraphics[width=16cm]{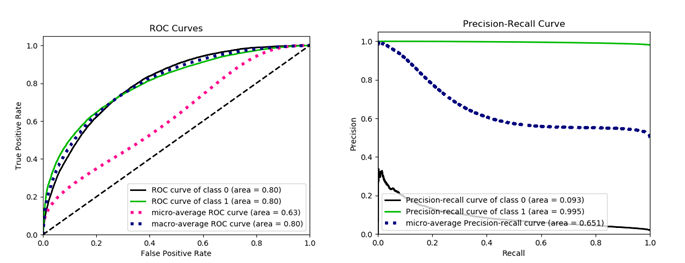}
\caption{ROC and PR on the out-of-time dataset for the final model}
\end{figure}

\chapter{Model Outputs, Reports and Uses}

\section{Model Outputs Overview}

This model is static but does support recalibration. Recalibration can be done by either retraining the model on the old data and the new data, the new data only, or just loading the current model then calling model.fit.

\section{Model Outputs}

This section includes the theory behind Shapley scores that can help understanding the model outputs. Possible uses are in providing consumers with a clear, concise, empirically-derived explanation as to why they were declined should the score be used for applicant decisioning.

\subsection{Explanation Models and Additive Feature Attribution Methods}

The ability to understand models and correctly interpret their predictions is crucial. It provides insight into how a model may be improved, supports understanding of the process being modelled and most importantly can help us answer regulatory and transparency requirements. While a wide variety of methods have been proposed to address the growing model complexity and the lack of their transparency, Shapley Explainers \citep{DBLP:journals/corr/LundbergL17} remain the clearest of all. 
Shapley Explainers \citep{DBLP:journals/corr/LundbergL17} introduce the idea of explaining a model through another model, called an explanation model. These models are defined as any interpretable approximation of the initial model.
Let $f$ be the original prediction model to be explained and $g$ the explanation model. Here we focus on local methods designed to explain a prediction $f(x)$ based on a single input $x$. Explanation models often use simplified inputs $x'$ that map the original inputs through a mapping function $x=h_x(x')$. Local methods try to ensure $g(z') \approx f(h_x(x'))$ whenever $z' \approx x'$. \citep{DBLP:journals/corr/LundbergL17} defines additive feature attribution methods as:

\begin{definition}{Additive Feature Attribution Methods}

Additive feature attribution methods have an explanation model that is a linear function of binary variables:

$$g(z') = \phi_0 + \sum_{i=1}^M \phi_i z'_i$$

Where $z' \in \{0,1\}^M$, $M$ is the number of simplified features and $\phi_i \in \mathbb{R}$

\end{definition}

Methods with explanation models matching Definition 1 attribute an effect $\phi_i$ to each feature, and summing the effects of all feature attributions approximates the output $f(x)$ of the original model.
There are three desirable properties we’d like our explanation model to have. The first one is local accuracy: when approximating the original model $f$ for a specific input $x$, local accuracy requires the explanation model to at least match the output of $f$ for the simplified input $x'$.

\begin{propriety}{Local Accuracy}

$$f(x) = g(x') = \phi_0 + \sum_{i=1}^M \phi_i x'_i$$

The explanation model $g(x')$ matches the original model $f(x)$ when $x = h(x')$, where $\phi_0 = f(h_x(0))$ represents the model output with all simplified inputs missing.

\end{propriety}

The second propriety is missingness. If the simplified inputs represent feature presence, then missingness requires features missing in the original input to have no impact.

\begin{propriety}{Missingness}

$$x'_i = 0 \Longrightarrow \phi_i = 0$$

Missingness constrains features where $x'_i = 0$ to have no attributed impact.

\end{propriety}

The third property is consistency. Consistency states that if a model changes so that some simplified input’s contribution increases or stays the same regardless of other inputs, that input’s attribution should not decrease.

\begin{propriety}{Consistency}

Let $f_x(z') = f(h_x(z'))$ and $z'/i$ denote setting $z'_i = 0$. For any two models $f$ and $f'$, if:

$$f'_x(z') - f'_x(z'/i) \geq f_x(z') - f_x(z'/i)$$

for all inputs $z' \in \{0,1\}^M$, then:

$$\phi_i(f',x) \geq \phi_i(f,x)$$

\end{propriety}

\subsection{Shapley Values}

A surprising attribute of the class of additive feature attribution methods is the presence of a single unique solution in this class with the three desirable properties above.

\begin{theorem}

Only one possible explanation model $g$ follows Definition 1 and satisfies properties 1,2 and 3:

$$\phi_i(f, x) = \sum_{z' \subseteq x'} \frac{|z'|! \, (M - |z'|- 1)!}{M!} \, (f_x(z') - f_x(z'/i))$$

where $|z'|$ is the number of non-zero entries in $z'$, and $z' \subseteq x'$ represents all vectors $z'$ where the non-zero entries are a subset of the non-zero entries in $x'$.

\end{theorem}

Theorem 2 follows from combined cooperative game theory results, where the values $\phi_i$ are known as Shapley values. Young (1985) \citep{Young1985} demonstrated that Shapley values are the only set of values that satisfy three axioms similar to Property 1, Property 3 and a final property that the authors have shown to be redundant in this setting. Property 2 is required to adapt the Shapley proofs to the class of additive distribution methods.

\subsection{Explainability}

SHAP (Shapley Additive exPlanations) is a unified approach to explain the output of any machine/deep learning model. SHAP connects the previous game theory and local additive explainers in an easy to use Python library that is well integrated with most algorithms in scikit-learn’s API, as well as XGBoost, LightGBM and CatBoost. 
Thanks to the shap package we can plot a summary graph that gives a simplified view of how certain feature values impact the output score of our model in terms of Shapley values.
The application of this method starts with the generation of a table of shap values for each variable in the model. We then explore the variables shap scores and understand how they interact to output an observations final score. Figure 5.1 shows the summary plot. It helps us get an overview of which features are most important for our model for every feature and for every sample. The plot sorts features by the sum of SHAP value magnitude over all samples, and uses SHAP values to show the distribution of the impacts each feature has on the model output (red high, blue low).

\begin{figure}[h!]
\centering
\includegraphics[width=8cm]{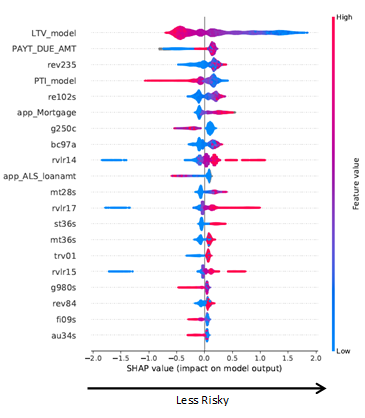}
\caption{Shapley Values summary for the final model}
\end{figure}

Dependence plots are plots that let us deep dive into a specific variable and explore how it interacts with other features: To understand how a single feature effects the output of the model we can plot the SHAP value of that feature VS. the value of the feature on all the examples in a dataset. Since SHAP values represent a feature’s responsibility for a change in the model output, Figure 5.2 below represents the change in the score as ContractState (not used in the final model) changes. Vertical dispersion at a single value of ContractState represents interaction effects with other features. 

\begin{figure}[ht]
\centering
\includegraphics[width=8cm]{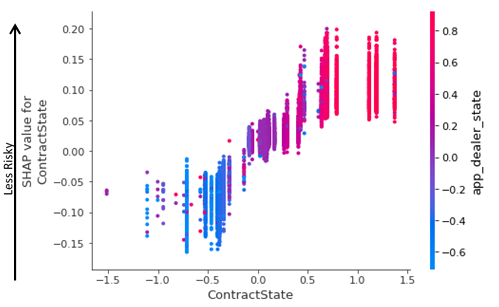}
\caption{Shapley Dependence Plot between the Contract State and the Dealer State}
\end{figure}

For a single observation, we can plot a force plot that shows how each contributing feature pushes the model output from the base value to the model output. Below in an example for a high score feature:

\begin{figure}[ht]
\centering
\includegraphics[width=12cm]{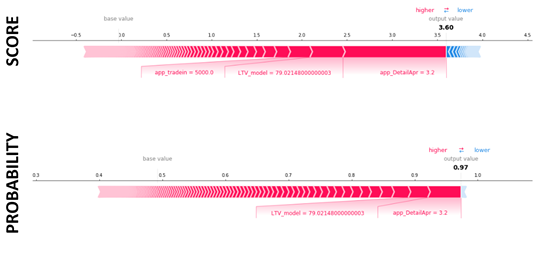}
\caption{Shapley Force plot for a high score observation}
\end{figure}

\begin{figure}[ht]
\centering
\includegraphics[width=12cm]{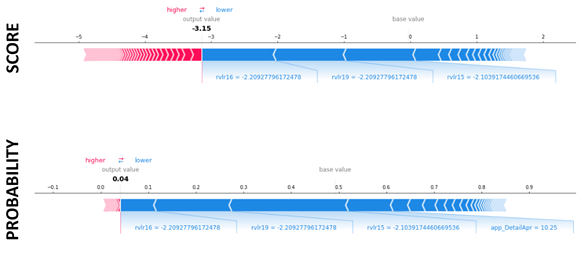}
\caption{Shapley Force plot for a low score observation}
\end{figure}

We can clearly see that a low APR (not used in the final model) and a low LTV are the main factors pushing the score up. The absence of negative impact features and the interaction effects are also a reason for this high score.
Below is an example of a low score observation. We can see that the main drivers in this case are the different bankcard Revolving/Transactor/Inactive patterns. Once again this is coherent with the SHAP summary plot.

\section{Model Reports}

\subsection{Swap Set Analysis}

The above tables show that our challenger model is able to detect more defaults without trading off a lot of good observations. 

\begin{table}[!h]
\begin{center}
\begin{minipage}{1.5in}
\begin{tabular}{lcc}
 & Our & BANK A \\
\hline
Worst 20\% & 6.00\% & 5.52\% \\
\hline
Total & 85967 & 85967 \\
\hline
Bads & 5157 & 4748 \\
\hline
Goods & 80810 & 81219 \\
 \end{tabular}
\end{minipage}
\quad \quad \quad \quad \quad \quad \quad
\begin{minipage}{1.5in}
 \begin{tabular}{lcc}
 & Our & BANK A \\
\hline
Worst 10\% & 8.26\% & 7.02\% \\
\hline
Total & 43637 & 43637 \\
\hline
Bads & 3561 & 3065 \\
\hline
Goods & 40121 & 40572 \\
 \end{tabular}    
\end{minipage}
\caption{Swap set analysis between our model and BANK A's}
\end{center}
\end{table}

\subsection{Process Flow Diagram}

Below is the process flow diagram describing how the model was built:

\begin{figure}[!h]
\centering
\includegraphics[width=14cm]{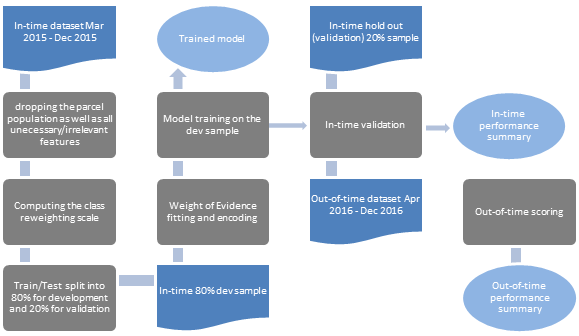}
\caption{Process Flow Diagram}
\end{figure}

\section{Summary of Key Assumptions and Limitations}

\subsection{Model Assumptions}

Following assumptions were made to prepare the data and make sure holds during model implementation:
\begin{itemize}
    \item Model was designed and developed using Auto loan origination data and hence should be used to rank order customers during loan origination only
    \item Model was developed using customer records who has a Fico 8 Auto above 660 and hence the implementation of model should focus on rank ordering customers within this Fico 8 Auto score range
    \item Usage of class reweighting means that the model does not output default probabilities. It can thus inly be used for rank ordering.
    \item Missing data holds no information beyond the fact that it is “missing” (see [6] for details about how the algorithm handles missingness). Any missing imputation done on original data might change the model output.
\end{itemize}

\subsection{Model Limitations}

\begin{itemize}
    \item Model is designed to only be applied to rank order customers at origination, not for account management
    \item Model can only be applied to customer with Fico 8 Auto range above 660
\end{itemize}

%%%%%APPENDIX%%%%%%%%

\begin{appendices}
\chapter{Proofs and Algorithms}

\section{Proof of the Error Estimation Bound}

\begin{proof}
  We remind the reader that Hoeffding’s inequality is as follows:
  
  Let $(X_i)_{1 \leq i \leq n}$ be n independent random variables such that for all $1 \leq i \leq n$:
  
  $$\mathbb{P}[a_i \leq X_i \leq b_i] = 1$$
  
  Where $a_i, \, b_i$ are real numbers such that $a_i \le b_i$.
  
  Hoeffding's inequality states that: 
  
  $$\mathbb{P}\left[\abs*{\sum_{i=1}^n X_i - \sum_{i=1}^n \mathbb{E}[X_i]} \ge t\right] \leq 2 \, \exp{\left ( \frac{-2t^2}{\sum_{i=1}^n (b_i - a_i)^2} \right )}$$
  
  Let’s consider the random variable $\mathbf{1}_{Y_i \neq h(X_i)}$. They are independent and bound between 0 and 1 by design. We can thus apply Hoeffding’s inequality:
  
  $$\mathbb{P}\left[\abs*{\sum_{i=1}^n \mathbf{1}_{Y_i \neq h(X_i)} - \sum_{i=1}^n \mathbb{E}[\mathbf{1}_{Y_i \neq h(X_i)}]} \ge t\right] \leq 2 \, \exp{\left ( \frac{-2t^2}{n} \right )}$$  
  
  Notice that $\mathbb{E}[\mathbf{1}_{Y_i \neq h(X_i)}] = L_{miss}(h(X_i))$. Therefore, by setting $t = \frac{n}{2}\sqrt{\frac{2}{n}\log\frac{2M}{\delta}}$ we have that:
  
  $$\mathbb{P}\left[\abs*{\sum_{i=1}^n \mathbf{1}_{Y_i \neq h(X_i)} - \sum_{i=1}^n \mathbb{E}[\mathbf{1}_{Y_i \neq h(X_i)}]} \ge \frac{n}{2}\sqrt{\frac{2}{n}\log\frac{2M}{\delta}}\right] \leq \delta$$
  
  Thus:
  
  $$1 - \mathbb{P}\left[2 \, \abs*{\hat L^n_{miss}(h) - L_{miss}(h)} \ge \sqrt{\frac{2}{n}\log\frac{2M}{\delta}}\right] \geq 1 - \delta$$
  
  The reader can easily prove that:
  
  $$L_{miss}(\hat h_H^n) - \min_{1 \leq j \leq M} L_{miss}(h_j) \leq 2 \sup_{h \in H} \abs*{\hat L^n_{miss}(h) - L_{miss}(h)}$$
  
  And finally by using the above inequality and switching to the complementary event’s probability:
  
  $$\mathbb{P} \left[L_{miss}\left ( h^n_H\right ) \leq \min_{1 \leq j \leq M} L_{miss} \left (h_j \right ) + \sqrt{\frac{2}{n}\log\left (\frac{2M}{\delta} \right )} \right] \geq 1 - \delta$$
  
  Q.E.D.

\end{proof}

\section{Forward Stagewise Additive Modelling}

\begin{algorithm}
\caption{Forward Stagewise Additive Modelling}\label{FSAM}
\begin{algorithmic}[1]
\Procedure{Initialize}{}
\State $f_0(x) \gets 0$
\BState \emph{For $m = 1$ to M}:
\State Compute $(\beta_m,\gamma_m) = \argmin_{\beta, \gamma} \sum_{i=1}^N L(y_i, f_{m-1}(x_i) + \beta \, b(x_i; \gamma))$.
\State $f_m(x) \gets f_{m-1}(x) + \beta_m \, b(x; \gamma_m)$
\EndProcedure
\end{algorithmic}
\end{algorithm}

\chapter{Class Imbalance}

\section{Why Imbalance can hurt your models}

Machine learning algorithms are built to minimize errors. Since the probability of instances belonging to the majority class is significantly high in imbalanced data set, the algorithms are much more likely to classify new observations to the majority class. For example, in a loan portfolio with an average default rate of 5\%, the algorithm has the incentive to classify new loan applications to non-default class since it would be correct 95\% of the time.

\begin{figure}[h!]
\centering
\includegraphics[width=8cm]{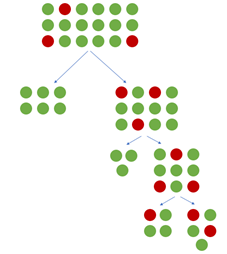}
\caption{How Imbalance Affects a Decision Tree}
\end{figure}

In the above-mentioned example, the cost of false negatives is significantly higher than that of a false positive, yet they are both penalized with a similar weight. Besides class reweighting that we presented above, several other approaches are possible.

\subsection{Data Resampling}

A possible alternative solution to class imbalance is to alter the dataset through oversampling (increasing the number of minority class members in the training set) or under-sampling (reduce the number of majority samples to balance the class distribution). The advantage of over-sampling is that no information from the original training set is lost, as all observations from the minority and majority classes are kept. On the other hand, Under-sampling might discard important information. But oversampling is prone to overfitting.
To solve this issue, Synthetic Minority Oversampling Technique (SMOTE) [9] aims to create new instances of the minority class by forming convex combinations of neighboring instances. This allows us to balance our data-set without as much overfitting, as we create new synthetic examples rather than using duplicates.

\begin{figure}[h!]
\centering
\includegraphics[width=10cm]{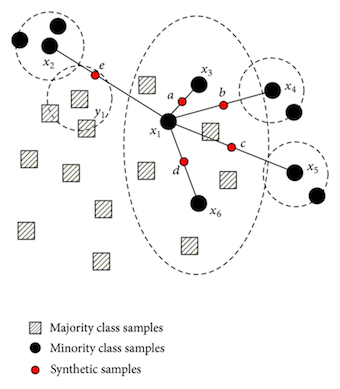}
\caption{SMOTE visualization}
\end{figure}

Advantages:
\begin{itemize}
    \item Allows generalization
    \item Adds new information
    \item No loss of information
\end{itemize}

Disadvantages:
\begin{itemize}
    \item Variance due to randomness
    \item Need to define the target percentage and the number of neighbors
    \item All examples are inside the convex hull
\end{itemize}

Below is a quick comparison between two logistic regressions: one fitted on imbalanced data and another one fitted on rebalanced data using SMOTE.

\begin{figure}[h!]
\centering
\includegraphics[width=16cm]{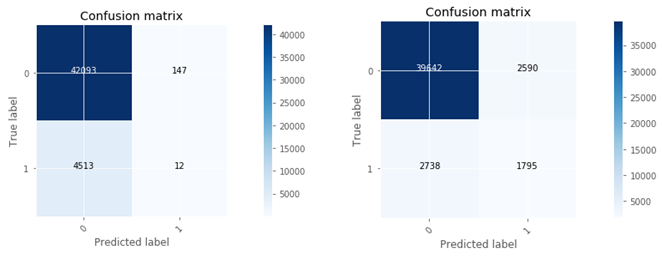}
\caption{Comparison of Logistic Regression on a dataset without and with SMOTE}
\end{figure}

\subsection{Appropriate Metrics}

In an imbalanced context, the only way to appropriately assess a model’s performance is to use AUC and PR at the same time. If we’re trying to compare two algorithms: one that tends to give a lot of false negatives and the other one doesn't. The number of FN only influences the FPR. Since we have class imbalance, it is likely that both algorithms will have a lot of True Negatives. Thus, it's pretty easy to achieve a small FPR even for the algorithm that gives a lot of False Negatives. In the below figure the model in black is clearly the best in terms of AUC but his low PR makes him sub-par compared to other models with slightly lower AUC.

\begin{figure}[h!]
\centering
\includegraphics[width=15cm]{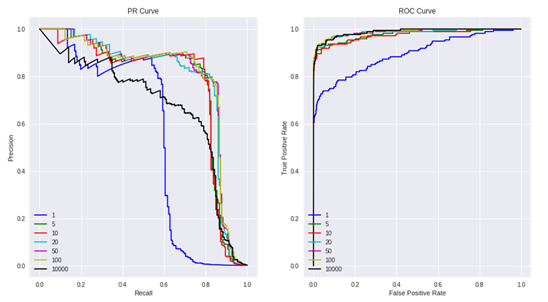}
\caption{Why PR is important along AUC}
\end{figure}

\end{appendices}

%%%%%BIB%%%%%%%%%%%%%
\newpage

\nocite{*}

\bibliographystyle{plain}
\bibliography{references}
\end{document}